\documentclass[12pt]{article}
\usepackage{a4wide}
\usepackage{amssymb}
\usepackage{amsmath}
\usepackage{graphicx}

\begin{document}
{\renewcommand{\thefootnote}{\fnsymbol{footnote}}
\begin{center}
{\LARGE  Faithful realizations of semiclassical truncations}\\
\vspace{1.5em}
Bekir Bayta\c{s},\footnote{e-mail address: {\tt bub188@psu.edu}}
Martin Bojowald\footnote{e-mail address: {\tt bojowald@gravity.psu.edu}}
and
Sean Crowe\footnote{e-mail address: {\tt stc151@psu.edu}}
\\
\vspace{0.5em}
Institute for Gravitation and the Cosmos,\\
The Pennsylvania State
University,\\
104 Davey Lab, University Park, PA 16802, USA\\
\vspace{1.5em}
\end{center}
}

\setcounter{footnote}{0}

\newtheorem{theo}{Theorem}
\newtheorem{lemma}{Lemma}
\newtheorem{defi}{Definition}

\newcommand{\proofend}{\raisebox{1.3mm}{\fbox{\begin{minipage}[b][0cm][b]{0cm}
\end{minipage}}}}
\newenvironment{proof}{\noindent{\it Proof:} }{\mbox{}\hfill \proofend\\\mbox{}}
\newenvironment{ex}{\noindent{\it Example:} }{\medskip}
\newenvironment{rem}{\noindent{\it Remark:} }{\medskip}

\newcommand{\case}[2]{{\textstyle \frac{#1}{#2}}}
\newcommand{\lP}{\ell_{\mathrm P}}

\newcommand{\md}{{\mathrm{d}}}
\newcommand{\tr}{\mathop{\mathrm{tr}}}
\newcommand{\sgn}{\mathop{\mathrm{sgn}}}

\newcommand*{\R}{{\mathbb R}}
\newcommand*{\N}{{\mathbb N}}
\newcommand*{\Z}{{\mathbb Z}}
\newcommand*{\Q}{{\mathbb Q}}
\newcommand*{\C}{{\mathbb C}}

\begin{abstract}
  Realizations of algebras in terms of canonical or bosonic variables can
  often be used to simplify calculations and to exhibit underlying
  properties. There is a long history of using such methods in order to study
  symmetry groups related to collective motion, for instance in nuclear shell
  models. Here, related questions are addressed for algebras obtained by
  turning the quantum commutator into a Poisson bracket on moments of a
  quantum state, truncated to a given order. In this application, canonical
  realizations allow one to express the quantum back-reaction of moments on
  basic expectation values by means of effective potentials. In order to match
  degrees of freedom, faithfulness of the realization is important, which
  requires that, at least locally, the space of moments as a Poisson manifold
  is realized by a complete set of Casimir--Darboux coordinates in local
  charts. A systematic method to derive such variables is presented and
  applied to certain sets of moments which are important for physical
  applications. If only second-order moments are considered, their
  Poisson-bracket relations are isomorphic to the Lie bracket of ${\rm
    sp}(2N,{\mathbb R})$, providing an interesting link with realizations of
  nuclear shell models.
\end{abstract}

\section{Introduction}

Semiclassical truncations approximate quantum dynamics by dynamical systems in
which expectation values are coupled to moments of a state. The classical
phase space is thereby extended to an enlarged manifold with a Poisson bracket
of expectation values and moments derived from the commutator of basic
operators. These canonical effective methods have been used in various
contexts, such as quantum chemistry \cite{QHDTunneling} and quantum cosmology
\cite{ROPP}, and they reproduce well-known results including tunneling
phenomena \cite{QHDHigher}, the low-energy effective action
\cite{EffAc,Karpacz}, or the Coleman--Weinberg potential \cite{CW}. However,
the enlargement of the classical phase space tends to complicate qualitative
interpretations as well as computations, in particular because moments, unlike
expectation values, do not form canonically conjugate pairs. In this paper, we
therefore analyze the problem of constructing canonical realizations of
Poisson systems, or their Casimir--Darboux coordinates. To second moment order
for a single pair of classical degrees of freedom, an interesting canonical
realization has been known for some time \cite{GaussianDyn,QHDTunneling}. Our
main goal is to extend these results to multiple degrees of freedom and to
higher orders in a semiclassical expansion.

At leading order, semiclassical truncations turn out to be closely related to
the Lie algebras ${\rm sp}(2N,{\mathbb R})$. Our methods and examples can
therefore be extended directly to finding canonical realizations for these
algebras. Moreover, once a canonical realization is found, one automatically
obtains a bosonic realization using the standard Poisson structure on the
complex numbers. (Canonical pairs are thereby replaced by classical analogs of
annihilation and creation operators.)

We put special emphasis on the construction of faithful realizations, in which
the number of independent variables is equal to the dimension of the original
system, and the co-rank of the Poisson tensor agrees with the number of
Casimir functions. Canonical and bosonic realizations of systems of the type
studied here have been used for several decades, but achieving faithfulness
often presented a problem. Bosonic realizations go back to theoretical work on
magnetic systems \cite{HP}. Interest in particular in bosonic realizations of
${\rm sp}(6,{\mathbb R})$ grew after the introduction of a symplectic model of
nuclear shells and vibrations \cite{Nuclearsp6}. Non-faithful bosonic
realizations have been used in several papers mainly to compute matrix
elements in irreducible representations
\cite{AlgebraicCollective,DynCollective,CoherentSymp,Bosonsp4,BosSymp}. Some
of these studies noted difficulties in finding faithful realizations, starting
with ${\rm sp}(4,{\mathbb R})$ \cite{Bosonsp4,BosSymp}.  Bosonic and canonical
realizations of Lie algebras other than ${\rm sp}(2N,{\mathbb R})$ have been
analyzed and formalized in \cite{Mukunda1,Mukunda2,Mukunda3,Rosen,Subreps},
which in most cases were not faithful.

Our results lead to an extension of some of the results of \cite{Bosonsp4} to
a faithful bosonic realization, but we expect the main applications of our
methods to be in semiclassical discussions of quantum mechanics. Even though
we address quantum systems, the use of semiclassical truncations means that we
are interested here in {\em classical} realizations of a system with Poisson
brackets. We do not consider the more complicated question of constructing
bosonic realizations of operator algebras --- the main topic of
\cite{Bosonsp4} --- in which factor ordering questions are relevant.

\section{Canonical Effective Methods}

Canonical effective equations \cite{EffAc,Karpacz} describe quantum effects
through interactions between expectation values and moments of a state with
respect to a fixed set of basic observables. The commutator of operators
induces a Poisson bracket on the space of expectation values and moments,
leading to an infinite-dimensional extension of the classical phase space. In
semiclassical approximations of varying degrees, finite-dimensional
truncations are used for each canonical pair. The Hamiltonian operator then
implies an effective Hamiltonian on the extended phase space for each of its
finite-dimensinal truncations, and quantum dynamics can be analyzed much like
a classical dynamical system.  Mathematically, canonical effective methods
replace partial differential equations for wave functions by a system of
coupled ordinary differential equations for an enlarged set of variables

We assume that the unital $*$-algebra ${\cal A}$ of observables defining the
quantum system is canonical, that is, generated by the unit operator together
with a finite set of self-adjoint position and momentum operators $Q_j$
and $\Pi_k$, $1\leq j,k\leq N$, with canonical commutation relations
\begin{equation}
 [Q_j,\Pi_k] = i\hbar \delta_{jk}\,.
\end{equation}
States are positive linear functionals $\omega$ from the algebra to the
complex numbers, such that $\omega(a^*a)\geq 0$ for all $a\in {\cal A}$
\cite{LocalQuant}. They may (but need not) be obtained from wave functions or
density matrices in or acting on a Hilbert space ${\cal H}$ on which ${\cal
  A}$ may be represented by $a\mapsto \hat{a}$: In such a case, every
$\psi\in{\cal H}$ defines a state $\omega_{\psi}\colon a\mapsto
\langle\hat{a}\rangle_{\psi}$, and every density matrix $\hat{\rho}$ defines a
state $\omega_{\rho}\colon a\mapsto {\rm tr}(\hat{a}\hat{\rho})$. To be
specific, and for easier comparison with the physics literature on the
subject, we will use the notation $\langle\hat{a}\rangle$ to denote
$\omega(a)$, but expectation values could as well be defined using mixed
states or algebraic states.

We introduce a set of basic variables taking real values:

\begin{defi} 
  Given a state on a canonical algebra ${\cal A}$ generated by self-adjoint
  $Q_j$ and $\Pi_k$, in addition to the unit, the {\em basic
    expectation values} are $q_j=\langle \hat{Q}_j\rangle\in{\mathbb R}$ and
  $\pi_k=\langle\hat{\Pi}_k\rangle\in{\mathbb R}$.

  For positive integers $k_i$ and $l_i$ such that $\sum_{i=1}^N(k_i+l_i)\geq
  2$, the {\em moments} of the state are given by
\begin{equation} \label{moments}
 \Delta\left(q_1^{k_1}\cdots q_N^{k_N}\pi_1^{l_1}\cdots \pi_N^{l_N}\right) =
 \langle (\hat{Q}_{1} - q_{1})^{{k_1}} \cdots (\hat{Q}_{N} - q_{N})^{{k_N}}
 (\hat{\Pi}_{1} - \pi_{1})^{{l_1}} \cdots (\hat{\Pi}_{N} - \pi_{N})^{{l_{N}}}
 \rangle_{\mathrm{Weyl}}  
\, , 
 \end{equation}
where the product of operators is Weyl (totally symmetrically) ordered.
\end{defi}

If the state is a Gaussian wave function in the standard Hilbert space on
which ${\cal A}$ can be represented, the moments obey the hierarchy 
\begin{equation} \label{hierarchy}
 \Delta\left(q_1^{k_1}\cdots q_N^{k_N}\pi_1^{l_1}\cdots \pi_N^{l_N}\right)=
 O\left(\hbar^{\frac{1}{2}\sum_n (l_n+k_n)}\right)\,.
\end{equation} 
This property motivates
\begin{defi}
 A state on a canonical algebra ${\cal A}$ is {\em semiclassical} if its
 moments obey the hierarchy (\ref{hierarchy}).
\end{defi}
A semiclassical state is much more general than the Gaussian family,
which has two free parameters per canonical pair of degrees of freedom. A
general semiclassical state, by contrast, allows for infinitely many free
parameters per canonical pair of degrees of freedom.

We will use the semiclassical hierarchy mainly in order to truncate the
infinite-dimensional space of expectation values and moments:
\begin{defi}
  The {\em semiclassical truncation of order $s\geq 2$} of a quantum system
  with canonical algebra ${\cal A}$ is a finite-dimensional manifold ${\cal
    P}_s$ with boundary, determined by global coordinates $q_j$, $\pi_k$ and
  all moments (\ref{moments}) such that $\sum_n(l_n+k_n)\leq s$. Its boundary
  components are obtained from the Cauchy--Schwarz inequality.
\end{defi}
A semiclassical truncation of order $s$ therefore includes variables up to
order $\frac{1}{2}s$ in $\hbar$ when evaluated on a Gaussian state.
Well-known components of the boundary are given by Heisenberg's uncertainty
principle
\begin{equation}
\Delta(q_j^2)\Delta(\pi_k^2)- \Delta(q_j \pi_k)^2 \geq \frac{\hbar^2}{4}
\delta_{jk} \, ,
\end{equation}
but there are higher-order versions relevant for $s>2$.

Basic expectation values and moments are equipped with a Poisson bracket
defined by
\begin{equation}
\{ \langle \hat{A} \rangle, \langle \hat{B} \rangle \}= \frac{1}{i \hbar}
\langle [\hat{A},\hat{B}] \rangle\, , 
\end{equation}
extended to all moments by using linearity and the Leibniz rule. The Poisson
bracket turns any semiclassical truncation into a phase space by ignoring in
$\{\Delta_1,\Delta_2\}$ all terms of order higher than $s$ in moments. This
condition includes the convention that the product of a moment of order $s_1$
and a moment of order $s_2$ is of semiclassical order $s_1+s_2$. Moreover, the
product of a moment of order $s_1$ with $\hbar^{s_2}$ is of order
$s_1+2s_2$. The consistency of this notion of order and the resulting
truncation has been shown in \cite{Counting}. 

In general, the Poisson tensor on a semiclassical truncation is not
invertible, such that there is no natural symplectic structure on a
semiclassical phase space. For instance, for $N=1$ the phase space of a
semiclassical truncation of order $s=1$ is five-dimensional with coordinates
$(q,\pi,\Delta(q^2),\Delta(q\pi),\Delta(\pi^2))$, and cannot be
symplectic. The non-zero basic brackets are
\begin{equation}
 \{q,\pi\} = 1
\end{equation}
and
\begin{equation} \label{qpiDelta}
 \{\Delta(q^2),\Delta(q\pi)\} = 2\Delta(q^2)\quad,\quad
 \{\Delta(q\pi),\Delta(\pi^2)\} = 2\Delta(\pi^2)\quad,\quad
 \{\Delta(q^2),\Delta(\pi^2)\} = 4\Delta(q\pi)\,.
\end{equation}

Quantum dynamics is determined by a Hamiltonian element $H\in{\cal A}$. We
assume that the Hamiltonian element is given by a sum of Weyl-ordered products
of the canonical generators. It defines the quantum Hamiltonian
$H_Q(\langle\cdot\rangle,\Delta)=\langle\hat{H}\rangle_{\langle\cdot\rangle,\Delta}$,
identified as a function of basic expectation values and moments through the
state used in $\langle\hat{H}\rangle$. On a semiclassical truncation of order
$s$, the quantum Hamiltonian leads to the effective Hamiltonian of order $s$,
\begin{eqnarray} \label{Heff}
H_{{\rm eff},s}
&=& \langle H(\hat{Q}_j + (\hat{Q}_j - q_j), \hat{\Pi}_k + (\hat{\Pi}_k -
\pi_k))\rangle \\ 
&=& H(q,\pi)+ \sum \limits_{\sum_n(j_n+k_n)=2}^{s}
 \frac{\partial^n H(q,\pi)}{\partial q_1^{j_1}\cdots \partial
  q_N^{j_N} \partial \pi_1^{k_1}\cdots \partial \pi_N^{k_N}}
\frac{\Delta\left(q_1^{j_1}\cdots q_N^{j_N}\pi_1^{k_1}\cdots
    \pi^{k_N}\right)}{j_1!\cdots j_N! k_1!\cdots k_N!} \, ,\nonumber
\end{eqnarray}
obtained by a formal Taylor expansion in $\hat{Q}_j - q_j$ and $\hat{\Pi}_k -
\pi_k$, where $H(q,\pi)$ is the classical Hamiltonian corresponding to
$H\in{\cal A}$. If the Hamiltonian is a polynomial in basic operators, the
expansion in (\ref{Heff}) is a finite sum and exact, and merely rearranges the
monomial contributions to $\hat{H}$ in terms of central moments.  By
definition of the Poisson bracket from the commutator, Hamiltonian equations
of motion
\begin{eqnarray}
\dot{f}(\langle\cdot\rangle,\Delta) = \{ f(\langle\cdot\rangle,\Delta), H_{{\rm
  eff},s}\}
\end{eqnarray}
generated by an effective Hamiltonian are truncations of Heisenberg's
equations of motion evaluated in a state.

\section{Faithful realizations of semiclassical truncations}

While the Poisson brackets $\{q_j,\pi_k\}=1$,
$\{q_j,\Delta\}=0=\{\pi_k,\Delta\}$ involving basic expectation values are
simple, the brackets between two moments are non-canonical and, in general,
non-linear \cite{EffAc,HigherMoments}:
\begin{eqnarray}\label{MomentBrackets}
\{\Delta(q^b\pi^a),\Delta(q^d\pi^c)\}&=&
a\, d \, \Delta(q^b\pi^{a-1}) \Delta(q^{d-1}\pi^c) - b c \Delta(q^{b-1}\pi^a)
\Delta(q^d\pi^{c-1})\nonumber\\
&&+\sum_{{\rm odd}\;n=1}^M
\left(\frac{i\hbar}{2}\right)^{n-1}
K_{abcd}^{n}\, \Delta(q^{b+d-n}\pi^{a+c-n})
\end{eqnarray}
with $M={\rm min}(a + c, b + d, a + b, c + d)$ and
\begin{equation}
K_{abcd}^n = \sum_{m= 0}^{n} (-1)^m m!(n-m)!
\left(\!\!\begin{array}{c} a\\m
\end{array}\!\!\right)
\left(\!\!\begin{array}{c} b\\n-m
\end{array}\!\!\right)
\left(\!\!\begin{array}{c} c\\n-m
\end{array}\!\!\right)
\left(\!\!\begin{array}{c} d\\m
\end{array}\!\!\right) \,.
\end{equation}
Since only odd $n$ are included in the sum in (\ref{MomentBrackets}), all
coefficients are real. Whenever a term $\Delta(q)$ or $\Delta(\pi)$ appears on
the right, it is understood to be zero, which is consistent with an extension
of (\ref{moments}) to $\sum (k_i+l_i)=1$ because $\langle\hat{a}-a\rangle=0$
for any operator $\hat{a}$. The brackets (\ref{MomentBrackets}) are therefore
linear in moments if and only if $a+b=2$ or $c+d=2$.

We will look for mappings of the moments to new variables such that the
Poisson brackets can be simplified. In particular, we will derive
canonical realizations of semiclassical truncations.
\begin{defi}
  A {\em canonical realization} of an algebra $(C^{\infty}(M),
  \{\cdot,\cdot\})$ on an open submanifold ${\cal U}\subset M$ is a
  homomorphism $(C^{\infty}({\cal U}),\{\cdot,\cdot\})\to (C^{\infty}({\mathbb
    R}^{2p}\times{\mathbb R}^I),\{\cdot,\cdot\}_{\rm can})$ to the algebra of
  functions on the Poisson manifold ${\mathbb R}^{2p+I}$ equipped with the
  canonical Poisson bracket on ${\mathbb R}^{2p}$, while $\{f,C\}_{\rm can}=0$
  for all $f\in C^{\infty}({\mathbb R}^{2p}\times{\mathbb R}^I)$ and $C\in
  {\mathbb R}^I$.

  A canonical realization of $(C^{\infty}(M), \{\cdot,\cdot\})$ is {\em
    faithful} if ${\rm dim}M=2p+I$ and $2p$ is equal to the rank of the
  Poisson tensor on $M$.
\end{defi}
Our examples of $M$ will be given by open submanifolds of the phase space of a
given semiclassical truncation. A closely related concept is that of a
bosonic realization:
\begin{defi}
  A {\em bosonic realization} of an algebra $(C^{\infty}(M), \{\cdot,\cdot\})$
  on an open submanifold ${\cal U}\subset M$ is a homomorphism
  $(C^{\infty}({\cal U}),\{\cdot,\cdot\})\to (C^{\infty}({\mathbb
    C}^p\times{\mathbb R}^I),\{\cdot,\cdot\}_{\rm bos})$ to the algebra of
  functions on the Poisson manifold ${\mathbb C}^p\times{\mathbb R}^I$, where
  ${\mathbb C}$ is equipped with the Poisson bracket $\{z^*,z\}_{\rm bos}=i$,
  while $\{f,C\}_{\rm bos}=0$ for all $f\in C^{\infty}({\mathbb
    C}^p\times{\mathbb R}^I)$ and $C\in {\mathbb R}^I$.

  A bosonic realization of $(C^{\infty}(M), \{\cdot,\cdot\})$ is {\em
    faithful} if ${\rm dim}M=2p+I$ and $2p$ is equal to the rank of the
  Poisson tensor on $M$.
\end{defi}
Pullbacks by the local symplectomorphisms
\begin{equation}
\Phi\colon {\mathbb R}^{2p}\to{\mathbb C}^p, (q_j,p_k)\mapsto
 \left({\textstyle\frac{1}{\sqrt{2}}}(q_l+ip_l)\right)
\end{equation}
define a bijection between canonical realizations and bosonic realizations
which preserves faithfulness. 

We note that the definitions impose reality conditions on the canonical or
bosonic variables. In particular, all $q_j$ and $p_k$ must be real, and a
bosonic pair $(z,z')$ with $\{z',z\}=i$ must be such that $z'=z^*$. 

\subsection{Poisson structure of semiclassical truncations}

Since basic expectation values have canonical Poisson brackets with one
another and zero Poisson brackets with any moment, the non-trivial
task is to construct a canonical realization of the space of moments for a
given semiclassical truncation, at fixed basic expectation values.
 
A canonical realization of a semiclassical truncation of order $s$ induces a
map 
\begin{equation}\label{mapping}
 {\cal X}^{(s)}\colon {\cal U}\subset {\cal P}_s \to {\mathbb
  R}^{2p}\times{\mathbb R}^I, (\Delta)\mapsto (s_{\alpha}, p_{\beta},
U_{\gamma})
\end{equation}
such that the variables $(s_{\alpha}, p_{\beta})$, $\{
s_{\alpha},p_{\beta} \} = \delta_{\alpha \beta}$, can be used as
coordinates on symplectic leaves defined by constant $U_{\gamma}$. The
coordinates $U_{\gamma}$ are therefore local expressions of Casimir functions
of the Poisson manifold \cite{Weinstein}.

A faithful realization requires a bijective map between the moments and
canonical variables. For a single degree of freedom and a semiclassical
truncation of order $s$, the dimension $D$ of the phase space is the number of
moments up to order $s$, or
\begin{equation}
D = \sum_{j=2}^s(j+1)= \frac{1}{2}(s^2 + 3 s - 4) \, .
\end{equation}
Note again that this dimension $D$ may be even or odd, depending on $s$. Even
if $D$ is even, the Poisson tensor is not guaranteed to be invertible.

Every function on a Poisson manifold we are considering can be expressed as
a function of finitely many moments $\Delta_i$ in some ordering.  We
introduce the Poisson tensor
\begin{equation}
\mathbb{P}_{ij}^{(s)}(\Delta) = \{ \Delta_i, \Delta_j \}\,,
\end{equation}
such that the Poisson brackets of the set of coordinates
$\mathcal{X}^{(s)}(\Delta)$ are
\begin{equation}
\{\mathcal{X}^{(s)}_{\alpha}(\Delta), \mathcal{X}^{(s)}_{\beta}(\Delta) \} =
\sum \limits_{i,j=1}^D
\frac{\partial\mathcal{X}^{(s)}_{\alpha}(\Delta)}{\partial \Delta_i} \,
\mathbb{P}_{ij}^{(s)}(\Delta) \,
\frac{\partial\mathcal{X}^{(s)}_{\beta}(\Delta)}{\partial \Delta_j} \,.
\end{equation}
The dimension of the nullspace of the Poisson tensor is equal to the number
of Casimir functions in a neighborhood of a given set of $\Delta_i$.

If the co-rank of the Poisson tensor is equal to $I$, at each point of 
phase space there exist $I$ linearly independent vectors ${\bf w}_{k}$,
$k=1,\ldots,I$ with components $w_{k}^i, i=1,\ldots,D$, such that
\begin{equation}
\sum \limits_{j=1}^D \mathbb{P}_{ij}^{(s)} w_{k}^j = 0, \,\,\, k= 1,\ldots, I \, .
\end{equation}
The vectors ${\bf w}_k=(w_{k}^j)$ are the eigenvectors of the Poisson tensor
with zero eigenvalue. Since this eigenspace has $I$-fold degeneracy, the ${\bf
  w}_{k}$ are not unique if $I>1$. They can be rearranged in linear
combinations with coefficients depending on $\Delta_i$.

Suppose one of the eigenvectors, ${\bf w}_{k}$, can be expressed as
\begin{equation}
w_{k}^{i} = \frac{\partial C_k(\Delta)}{\partial \Delta_{i}}\,.
\end{equation}
Then $C_k(\Delta)$ is a Casimir function which commutes with any function on
the Poisson manifold.  At a given point, each 1-form ${\rm d}C_k$ defines a
smooth submanifold of codimension one in the Poisson manifold through ${\rm
  d}C_k=0$. As the eigenvectors ${\bf w}_k$, and therefore the ${\rm d}C_k$,
are linearly independent, the intersections of all $I$ ($D-1$)-dimensional
submanifolds is a ($D-I$)-dimensional submanifold, called a symplectic
leaf. If we choose local coordinates $(v_1,\ldots,v_{D-I})$ on a symplectic
leaf, we have $(v_1,\cdots,v_{2n},C_1,\cdots,C_I)$ as a coordinate system on
phase space, where $n=\frac{1}{2}(D-I)$. The Poisson tensor in these
coordinates takes the form
\begin{equation} \label{Pij}
\mathbb{P}^{(s)}_{ij} = \left(
\begin{array}{c|c}
\tilde{\mathbb{P}}^{(s)}_{\alpha \beta} & 0  \\
\hline
0 & 0
\end{array}
\right) \, ,
\end{equation}
where $\tilde{\mathbb{P}}^{(s)}_{\alpha \beta} = \{ v_{\alpha},v_{\beta}
\}$ and $\mathrm{det}(\tilde{\mathbb{P}}^{(s)}_{\alpha \beta}) \neq 0$. A
faithful canonical realization provides a map
\begin{equation}
(v_1,\cdots,v_{2n},C_1,\cdots ,C_I) \rightarrow (s_1,\cdots , s_n, p_1,
\cdots , p_n, U_1, \cdots, U_I) 
\end{equation}
of the local coordinates. After applying this map, the Poisson tensor has the
form (\ref{Pij}) with
\begin{equation}
\tilde{\mathbb{P}}_{\alpha\beta}^{(s)} = \left(\begin{array}{cc}
0 & \mathbb{I}_n \\
- \mathbb{I}_n & 0
\end{array}\right) \, .
\end{equation}
Darboux' theorem shows that local canonical coordinates $s_{\alpha}$ and
$p_{\beta}$ exist.

As $\dot{C}_I = \{C_I,H\} = 0$ for any Hamiltonian $H$, motion is always
confined to a symplectic leaf $C_I=\mathrm{const}$. Moreover, the existence of
a Casimir function implies that the Hamiltonian is not unique because
$\{f,H\}=\{f,H+\lambda^IC_I\}$ for any phase-space function $f$ and
$\lambda^I\in{\mathbb R}$.

\subsection{Algebraic structure of second-order semiclassical truncations} \label{s:Alg}

For a system with $N$ classical degrees of freedom, we collectively refer to
$q_j$ and $\pi_k$ as $x_i$, $i=1,\ldots,2N$.  As can be seen from
(\ref{MomentBrackets}) or directly from commutators, the Poisson brackets of
second-order semiclassical truncations are then of the form
\begin{equation} \label{DeltaDelta} 
\left\{\Delta(x_i x_j),\Delta(x_k
    x_l)\right\}=\sum_{m\leq n}f_{ij;kl}^{m n}\Delta(x_mx_n)\,.
\end{equation}
The $\Delta(x_ix_j)$ form an independent set of moments if we require that
$i\leq j$.
 
The brackets are linear and form a Lie algebra with structure constants
\begin{equation}
f_{ij;kl}^{m
  n}=\tau_{ik}\delta_{j}^m\delta_{l}^n+\tau_{il}\delta_{j}^m\delta_{k}^n+
\tau_{jk}\delta_{i}^m\delta_{l}^n+\tau_{jl}\delta_{i}^m\delta_{k}^n\,, 
\end{equation}
using $\tau_{ij}=\left\{x_i,x_j\right\}$. For $\tau_{ij}$, we have the
identity
\begin{equation}
 \sum_j\tau_{ij}\tau_{jk}= \sum_j\{x_i,x_j\}\{x_j,x_k\}=-\delta_{ik}
\end{equation}
because both brackets are non-zero if and only if $x_j$ is canonically
conjugate to both $x_i$ and $x_k$, which implies $x_i=x_k$ for basic
variables.  We note that the $f_{ij;kl}^{mn}$ are manifestly symmetric in the
index pairs $(i,j)$ and $(k,l)$, but not in $(m,n)$.

Instead of summing over restricted
double indices, it is more convenient to symmetrize all of them explicitly, in
particular 
\begin{equation}
 f_{ij;kl}^{(m
   n)}=\frac{1}{2}\left(
\tau_{ik}\delta_{j}^m\delta_{l}^n+\tau_{il}\delta_{j}^m\delta_{k}^n+ 
\tau_{jk}\delta_{i}^m\delta_{l}^n+\tau_{jl}\delta_{i}^m\delta_{k}^n+
\tau_{ik}\delta_{j}^n\delta_{l}^m+\tau_{il}\delta_{j}^n\delta_{k}^m+  
\tau_{jk}\delta_{i}^n\delta_{l}^m+\tau_{jl}\delta_{i}^n\delta_{k}^m\right)\,, 
\end{equation}
and include all $\Delta(x_mx_n)$ in (\ref{DeltaDelta}) using
$\Delta(x_mx_n)=\Delta(x_nx_m)$. Summations over restricted double indices
$(m,n)$ such that $m\leq n$ can then be replaced by two full summations over
$m$ and $n$. For instance,
\begin{equation}
 \left\{\Delta(x_i x_j),\Delta(x_k
    x_l)\right\}=\sum_{m\leq n}f_{ij;kl}^{m
    n}\Delta(x_mx_n)=\sum_{m,n}f_{ij;kl}^{(m n)}\Delta(x_mx_n)\,. 
\end{equation}

\subsubsection{Cartan metric and root vectors}

We compute the Cartan metric
\begin{equation}\label{Cartan}
 g_{ij;kl}=\sum_{m,n,o,p} f_{ij;mn}^{(op)} f_{kl;op}^{(mn)} 
= 4(N+1)\left(\tau_{i l}\tau_{kj}+\tau_{i k}\tau_{lj}\right)\,.
\end{equation}

\begin{lemma}
The Cartan metric (\ref{Cartan}) is non-degenerate.
\end{lemma}
\begin{proof}
  The metric acts on objects of the form
  $V=\sum_{i,j}V^{ij}\Delta(x_i x_j)$ via
\begin{equation}
 g(V_1,V_2)= \sum_{i,j,k,l}g_{ij;kl} V_1^{ij}V_2^{kl}\,.
\end{equation}
For $V$ to be non-zero we need
$\text{Sym}(V^{ij})=\frac{1}{2}(V^{ij}+V^{ji})\neq 0$ because $\Delta(x_i
x_j)=\Delta(x_jx_i)$. Suppose there is a non zero object $V$ in the null space
of $g$, such that $g\left(V,\cdot\right)=0$ or
$\sum_{i,j}V^{ij}g_{ij;kl}=0$. Using (\ref{Cartan}) and rearranging, we find
\begin{equation}\label{con}
 0=8(N+1)\sum_{i,j}\tau_{li}\,\text{Sym}(V^{ij})\,\tau_{jk}\,.
\end{equation}
Because $\tau$ is invertible, (\ref{con}) implies that $V^{ij}$ is
antisymmetric, but then $V=0$. We conclude that $g$ is non-degenerate.
\end{proof}

The algebra of second-order moments is therefore a semi-simple Lie algebra.
We can show that it is actually simple, and identify it, by examining its
Dynkin diagram. We should first find the Cartan subalgebra.
\begin{lemma}
 The adjoint action of any moment of the form $\Delta(q_iq_j)$,
 $\Delta(\pi_i\pi_j)$, or $\Delta(q_k\pi_l)$ with $k\not=l$ is nilpotent.
\end{lemma}
\begin{proof}
  The claim is easy to see for $\Delta(q_iq_j)$ and $\Delta(\pi_i\pi_j)$: The
  adjoint action of $\Delta(q_iq_j)$ on a moment $\Delta$ is a sum of moments
  in which any $\pi_k$ that may appear in $\Delta$ is replaced by $q_k$, if
  $k=i$ or $k=j$. After applying this action twice, no $\pi_k$ is left and the
  third application gives zero. Analogous arguments hold for
  $\Delta(\pi_i\pi_j)$.

 For $\Delta(q_k\pi_l)$ with $k\not=l$, the adjoint action is non-zero only on
 moments of the form $\Delta(x\pi_k)$ or $\Delta(yq_l)$, where $x$ and $y$ can
 be any position or momentum component. In the first case, we compute
\begin{eqnarray*}
 \{\Delta(q_k\pi_l),\Delta(x\pi_k)\} &=& \Delta(x\pi_l)+
 \{\pi_l,x\}\Delta(q_k\pi_k)\\ 
 &=& \left\{\begin{array}{cl} \Delta(q_l\pi_l)- \Delta(q_k\pi_k) &\mbox{if
     }x=q_l\\ \Delta(x\pi_l) & \mbox{if }x\not=q_l\end{array}\right.
\end{eqnarray*}
Therefore,
\begin{eqnarray*}
 \{\Delta(q_k\pi_l),\{\Delta(q_k\pi_l),\Delta(x\pi_k)\}\} &=&
 \left\{\begin{array}{cl} 
     -\Delta(q_k\pi_l) &\mbox{if }x=q_l\\ \{q_k,x\}\Delta(\pi_l^2) \mbox{if
     }x\not=q_l\end{array}\right.\\
 &=& \left\{\begin{array}{cl}
     -\Delta(q_k\pi_l) &\mbox{if }x=q_l\\ \Delta(\pi_l^2)& \mbox{if
     }x=\pi_k\\0 &\mbox{otherwise}\end{array}\right.
\end{eqnarray*}
The next adjoint action of $\Delta(q_k\pi_l)$ gives zero, and similarly on
$\Delta(yq_l)$.
\end{proof}
Since nilpotent actions are non-diagonalizable, we construct the Cartan
subalgebra from moments of the form $\Delta(q_i \pi_i)$. Since they
Poisson commute with one another, they span the Cartan subalgebra
\begin{equation}
H=\left\langle\Delta(q_i \pi_i)\right\rangle_{1\leq i\leq N}\,.
\end{equation}
The moments $\Delta(q_i \pi_i)$ are orthogonal to one another and have the same
norm with respect to the Cartan metric.

The entire set of moments forms a Cartan--Weyl basis.  For any $\Delta(q_i
\pi_i)$, the set of basic moments $\Delta(x_kx_l)$ with $k\leq l$ is an
eigenbasis of the adjoint action with eigenvalues $2$ if $x_k=x_l=\pi_i$, $1$ if
$x_l=\pi_i$ and $q_i\not=x_k\not=\pi_i$, $-1$ if $x_k=q_i$ and
$q_i\not=x_l\not=\pi_i$, $-2$ if $x_k=x_l=q_i$, and zero otherwise. The
eigenvectors with eigenvalues $\pm 2$ have eigenvalue $0$ with any other
$\Delta(q_i\pi_i)$, while the eigenvectors with eigenvalues $\pm 1$ are shared
by two moments of the form $\Delta(q_i \pi_i)$. The root system is therefore
given by all vectors with only two non-zero components of opposite sign and
absolute value one, and vectors with a single non-zero component equal to $\pm
2$.  A suitable subset of eigenmoments with the smallest possible positive
eigenvalues for the adjoint action of all $\Delta(q_i\pi_i)$ gives the simple
root vectors
\begin{equation}
\left\{\Delta(q_2\pi_1),\Delta(q_3\pi_2),\ldots,\Delta(q_{N}\pi_{N-1}),
\Delta(\pi_N^2)\right\}\,,   
\end{equation}
with simple roots
\begin{equation}
 \left(\begin{array}{c}1\\-1\\0\\0\\\vdots\\0\\0\\0\end{array}\right)
\quad,\quad   
 \left(\begin{array}{c}0\\1\\-1\\0\\\vdots\\0\\0\\0\end{array}\right)
\quad,\;\ldots\;,\quad
 \left(\begin{array}{c}0\\0\\0\\0\\\vdots\\0\\1\\-1\end{array}\right)
\quad,\quad 
 \left(\begin{array}{c}0\\0\\0\\0\\\vdots\\0\\0\\2\end{array}\right)\,.
\end{equation}
The resulting Dynkin diagram, shown in Fig.~\ref{f:Dynkin}, belongs to ${\rm
  sp}(2N,{\mathbb R})$.

\begin{figure}
\begin{center}
\includegraphics[scale=0.5]{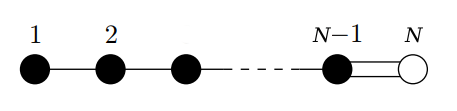}
\caption{The Dynkin diagram for a second-order semiclassical truncation. We
  adopt the convention that the filled circles correspond to shorter roots and
  the empty circles correspond to longer roots.  \label{f:Dynkin}} 
\end{center}
\end{figure}

The Casimir functions of ${\rm sp}(2N,{\mathbb R})$ can therefore be thought
of as approximate constants of motion in quantum mechanics: At the second
semiclassical order, the Hamiltonian is a function of basic expectation values
and second-order moments, and the ${\rm sp}(2N,{\mathbb R})$ Casimir functions
commute with any such function.  These constants of motion can be written as
\begin{equation}
U_{2m}\propto \tr{\left[\left(\tau \Delta\right)^{2m}\right]}\,\, , \,\,m\leq N
\end{equation}
where $\Delta$ is a matrix with components $\Delta_{ij}=\Delta(x_ix_j)$, and
$\tau_{ij}=\{x_i,x_j\}$ as before. There is one approximate constant of motion
per classical degree of freedom.

\subsubsection{Example of ${\rm sp}(4,{\mathbb R}$)}

For two classical degrees of freedom, we show the Cartan metric
ordering the moments as
\begin{equation}
 \left\{\Delta(\pi_1^2),\Delta(\pi_1
   q_1),\Delta(q_1^2),\Delta(\pi_2^2),\Delta(\pi_2 
   q_2),\Delta(q_2^2),\Delta(\pi_1 \pi_2),\Delta(\pi_1 q_2),\Delta(\pi_2
   q_1),\Delta(q_1 q_2)\right\}\,.
\end{equation}
The result,
\begin{equation}
g=\left(
\begin{array}{cccccccccc}
    0 & 0& -24& 0 & 0 & 0 & 0& 0& 0 & 0 \\
    0 & 12& 0& 0 & 0 & 0 & 0& 0& 0 & 0 \\
    -24 & 0& 0& 0 & 0 & 0 & 0& 0& 0 & 0 \\
    0 & 0& 0& 0 & 0 & -24 & 0& 0& 0 & 0 \\
    0 & 0& 0& 0 & 12 & 0 & 0& 0& 0 & 0 \\
    0 & 0& 0& -24 & 0 & 0 & 0& 0& 0 & 0 \\
    0 & 0& 0& 0 & 0 & 0 & 0& 0& 0 & -12 \\
    0 & 0& 0& 0 & 0 & 0 & 0& 0& 12 & 0 \\
    0 & 0& 0& 0 & 0 & 0 & & 12& 0 & 0 \\
    0 & 0& 0& 0 & 0 & 0 & -12& 0& 0 & 0
\end{array}\right)\,,
\end{equation}
is easily seen to be non-degenerate. The Cartan subalgebra is
\begin{equation}
H=\left\langle\Delta(q_1 \pi_1),\Delta(q_2 \pi_2)\right\rangle\,,
\end{equation}
and the simple root vectors
\begin{equation}
\left\{\Delta(q_2\pi_1),\Delta(\pi_2^2)\right\}
\end{equation}
imply simple roots
\begin{equation}
\alpha_1= \left(\begin{array}{c} 1\\-1 \end{array}\right)\quad,\quad
\alpha_2= \left(\begin{array}{c} 0\\2 \end{array}\right)
\end{equation}
corresponding to the Cartan matrix
\begin{equation}
K=\left(\begin{array}{cc} 2 & -1 \\
    -2 & 2 \end{array}\right)
\end{equation}
of ${\rm sp}(4,{\mathbb R})$ (or $C_2$).

\subsection{Examples}

We present standard examples of faithful realizations before we
proceed with the general theory.

\subsubsection{The Lie algebra su(2)}

The Poisson bracket for su(2) with generators $S_i$, $i=1,2,3$, is
given by
\begin{equation}
\left\{S_i,S_j\right\}=\sum_{k=1}^3\epsilon_{ijk}S_k \, .
\end{equation}
It is well known that $S^2=\sum_{i=1}^3S_i^2$ is a Casimir function of this
algebra. The task is to find a pair of functions of the generators that are
canonically conjugate with respect to the original Poisson tensor. These
variables can be defined implicitly by
\begin{equation}
S_x=\sqrt{S^2-S_z^2}\,\,\mathrm{cos}(\phi)\quad,\quad
S_y=\sqrt{S^2-S_z^2}\,\,\mathrm{sin}(\phi)\,,
\end{equation}
such that $\left\{\phi,S_z \right\}=1$. Solving for $\phi$ and inserting it
into the Poisson bracket, we indeed have
\begin{equation}
\left\{\phi,S_z\right\}=\left\{\mathrm{arctan}
  \left(S_y/S_x\right),S_z\right\}= \frac{\partial
  \arctan(S_y/S_x)}{\partial S_x} \{S_x,S_z\}+  \frac{\partial
  \arctan(S_y/S_x)}{\partial S_y} \{S_y,S_z\} =1 \, .
\end{equation}

\subsubsection{The Lie algebra su(1,1)}

The Lie algebra ${\rm su}(1,1)$ is defined by the relations
\begin{equation} \label{K}
[\mathcal{K}_0,\mathcal{K}_1]=-\mathcal{K}_2\quad,\quad
[\mathcal{K}_1,\mathcal{K}_2]=\mathcal{K}_0\quad,\quad
[\mathcal{K}_0,\mathcal{K}_2]=\mathcal{K}_1 \, .
\end{equation}
For this bracket, a faithful canonical realization is given by
\begin{equation}\label{K0Can}
\mathcal{K}_0=k+\frac{1}{2}\left(s^2+p_s^2\right)\quad,\quad
\mathcal{K}_1=\frac{s}{2}\sqrt{4k +s^2+p_s^2}\quad,\quad
\mathcal{K}_2=\frac{p_s}{2}\sqrt{4k +s^2+p_s^2}\,,
\end{equation}
where $K_1^2 + K_2^2 - K_0^2 = -k^2$ is the Casimir function and $s$ and $p_s$
are canonically conjugate variables. 

\subsubsection{The Lie algebra ${\rm sp}(2,{\mathbb R})$}

The Lie algebra ${\rm sp}(2,{\mathbb R})$ can be expressed as the set of
matrices of the form $\left(\begin{array}{cc}c&a\\b&-c\end{array}\right)$,
with generators
\begin{equation}
 A=\left(\begin{array}{cc}0&1\\0&0\end{array}\right)\quad,\quad B=
 \left(\begin{array}{cc}0&0\\1&0\end{array}\right) \quad,\quad
 C=\left(\begin{array}{cc}1&0\\0&-1\end{array}\right)
\end{equation}
and relations
\begin{equation}
 [A,B]=C\quad,\quad [A,C]=-2A\quad,\quad [B,C]=2B\,.
\end{equation}
Over the complex numbers, this Lie algebra
is isomorphic to ${\rm su}(1,1)$ via
\begin{equation}
 A=K_2+iK_1\quad,\quad B=K_2-iK_1\quad,\quad C=2iK_0\,.
\end{equation}
The canonical realization (\ref{K0Can}) can therefore be mapped
to this case:
\begin{equation} \label{ABC}
 A=\frac{1}{2}(p_s+is)\sqrt{4k+s^2+p_s^2}\quad,\quad
 B=\frac{1}{2}(p_s-is)\sqrt{4k+s^2+p_s^2}\quad,\quad C=i(2k+s^2+p_s^2)\,.
\end{equation}
However, because ${\rm sp}(2,{\mathbb R})$ and ${\rm su}(1,1)$ are different
real forms, these generators are not real. The generators (\ref{ABC})
therefore do not present a suitable canonical realization for our purposes.

Similarly, using $b=2^{-1/2}(s+ip_s)$, we obtain generators
\begin{eqnarray}  \label{ABCBos}
A =ib^*\sqrt{b^* b + 2k}\quad, \quad B=-ib\sqrt{b^* b +
  2k}\quad,\quad C = 2 i(b^* b + k)
\end{eqnarray}
of Holstein--Primakoff type \cite{HP} in which $A$ and
$B=A^*$ can be quantized to raising and lowering operators. However, these
generators are not real either, and do not present a suitable bosonic
realization. 

\subsubsection{Second-order semiclassical truncation for a single pair of
  classical degrees of freedom}

The constructions used in \cite{GaussianDyn,QHDTunneling} can be interpreted
as a faithful canonical realization
\begin{equation} \label{MomentCanReal}
 \Delta(q^2)=s^2\quad,\quad \Delta(q\pi)=sp_s\quad,\quad \Delta(\pi^2)=p_s^2+
 \frac{U}{s^2}
\end{equation}
of a semiclassical truncation with $N=1$, $s=2$, and Casimir function $U$. 

The mapping 
\begin{equation} \label{ABCDelta}
 A=-\frac{1}{2}\Delta(\pi^2)\quad,\quad B=\frac{1}{2}\Delta(q^2)\quad,\quad
 C=\Delta(q\pi) 
\end{equation}
generates an isomorphism to ${\rm sp}(2,{\mathbb R})$, giving a simple example
of the results of Section~\ref{s:Alg}, and a corresponding faithful canonical
realization of ${\rm sp}(2,{\mathbb R})$. If we use the canonical realization
(\ref{K0Can}) of ${\rm su}(1,1)$, on the other hand, we obtain complex
expressions for the moments and therefore violate the reality conditions
imposed on faithful canonical realizations.

Using (\ref{ABCDelta}), the canonical realization (\ref{MomentCanReal}) can be
related to (\ref{ABCBos}) if we define
\begin{equation}
b' = \frac{-\sqrt{2}iA}{\sqrt{-iC+2k}}=
\frac{i}{\sqrt{2}}\frac{p_s^2+U/s^2}{\sqrt{\sqrt{U}-i s p_s}}\quad,\quad 
b = \frac{\sqrt{2}iB}{\sqrt{-iC+2k}}=\frac{i}{\sqrt{2}}
\frac{s^2}{\sqrt{\sqrt{U}-i s p_s}} 
\end{equation}
with $U=4k^2$, such that $\{b',b\}=i$. However, reality conditions are again
violated because $b'\not=b^*$.

\subsubsection{Non-faithful bosonic realization of ${\rm sp}(2N,{\mathbb R})$}

The Lie algebra ${\rm sp}(2N,{\mathbb R})$ can be written with $N(2N+1)$
generators $A_{ij}$ ($i\leq j$), $B_{ij}$ ($i\leq j$) and $C_{ij}$ where
$i,j=1,\ldots,N$ and relations \cite{DynCollective}
\begin{eqnarray} \label{ABCRel}
 &&[A_{ij},A_{i'j'}]=0=[B_{ij},B_{i'j'}]\\
 && [B_{ij},A_{i'j'}] = C_{j'j}\delta_{ii'}+ C_{i'j}\delta_{ij'}+
 C_{j'i}\delta_{ji'}+ C_{ii'}\delta_{jj'}\\
 && [C_{ij},A_{i'j'}]= A_{ij'}\delta_{ji'}+A_{ii'}\delta_{jj'}\\
 && [C_{ij},B_{i'j'}]= -B_{jj'}\delta_{ii'}- B_{ji'}\delta_{ij'}\\
 && [C_{ij},C_{i'j'}]= C_{ij'}\delta_{i'j}- C_{i'j}\delta_{ij'}\,.
\end{eqnarray}
It has a bosonic realization \cite{Collective,DynCollective,Bosonsp4,BosSymp}
\begin{equation} \label{NonFaithful}
 A_{ij} = \sum_{\alpha=1}^n b_{i\alpha}^*b_{j\alpha}^*\quad,\quad
 B_{ij} = \sum_{\alpha=1}^n b_{i\alpha}b_{j\alpha} \quad,\quad
 C_{ij} = \frac{1}{2}\sum_{\alpha=1}^n \left(b^*_{i\alpha}b_{j\alpha}+
   b_{j\alpha}b^*_{i\alpha}\right)
\end{equation}
for every integer $n\geq 1$, with $nN$ boson variables $b_{i\alpha}$ (implying
$2nN$ degrees of freedom). 

For our purposes, this realization violates reality conditions. Moreover, it
is not faithful: Since $2N+1$ is odd, the number of degrees of freedom cannot
match the dimension $N(2N+1)$ of ${\rm sp}(2N,{\mathbb R})$, and since ${\rm
  sp}(2N,{\mathbb R})$ has rank $N$, it has $N$ Casimirs. For a faithful
bosonic realization, one therefore needs $N^2$ boson variables $b_{i\alpha}$
(that is, $n=N$) and $N$ Casimir variables. Finding an explicit realization of
this form has proven to be difficult even for ${\rm sp}(4,{\mathbb R})$. For
instance, possible expressions have been given up to solving complicated
partial differential equations \cite{Bosonsp4} or diagonalizing large matrices
\cite{BosSymp}. In the next section, we will solve this problem for the
analogous question of finding a faithful canonical realization of a
second-order semiclassical truncation with two classical degrees of freedom,
which is algebraically equivalent to ${\rm sp}(4,{\mathbb R})$.

\section{Constructing Casimir--Darboux coordinates}

A partially constructive proof of Darboux' theorem for symplectic manifolds is
presented in \cite{Arnold}: Given a symplectic manifold $(M,\omega)$, the
following steps demonstrate the existence of Darboux coordinates $(q_j,\pi_k)$
in a neighborhood ${\cal U}\subset M$ around a given point $x\in M$, such that
$\omega=\sum_j {\rm d}q_j\wedge {\rm d}\pi_j$.  We first choose some function
on $M$, calling it $q_1$, such that ${\rm d}q_1\not=0$ at $x$. Its Hamiltonian
vector field $X_{q_1}$ is then non-zero and generates a non-trivial flow
$F_{q_1}(t)=\exp(tX_{q_1})$ in a neighborhood of $x$. Choosing a hypersurface
transverse to $X_{q_1}$, we can endow the whole neighborhood with a pair of
coordinates given by $q_1$ and $\pi_1=-t$, defined by the parameter $t$ of the
Hamiltonian flow such that $t=0$ on the hypersurface. These two coordinates
are canonically conjugate because
\begin{equation}
 \{q_1,\pi_1\}= X_{q_1}t=\frac{\partial}{\partial t}t=1\,.
\end{equation}

We then move on to the hypersurface defined by $q_1=0=\pi_1$, apply the
previous steps, and iterate until we have the required number of coordinates
$q_j$ and $\pi_k$ defined on a family of hypersurfaces of decreasing
dimension. Starting with the last hypersurface of dimension two, we
iteratively transport the coordinates into a neighborhood within the next
higher hypersurface by declaring that they take constant values on all lines
of the flows $F_{q_i}(s)F_{\pi_i}(t)$, if $q_i$ and $\pi_j$ have already
been transported in this way. The proof concludes by showing that the
coordinates transported to the neighborhood ${\cal U}$ of $x$ in $M$ are
indeed canonical.

The steps used to prove Darboux' theorem for symplectic manifolds can be
simplified and extended to a systematic procedure to derive Casimir--Darboux
coordinates on Poisson manifolds. We keep the first step, but instead of using
hypersurfaces of constant canonical coordinates we construct hypersurfaces
which are Poisson orthogonal to the already constructed canonical pairs. This
modification eliminates the need to transport coordinates from hypersurfaces
to the full manifold. We first illustrate the method for the second-order
semiclassical truncation of a single pair of classical degrees of freedom.

\subsection{Canonical realization for a single pair of degrees of freedom
  at second order}

The Poisson brackets of our non-canonical coordinates $\Delta(q^2)$,
$\Delta(q\pi)$ and $\Delta(\pi^2)$ are given in (\ref{qpiDelta}):
\begin{equation}
  \{\Delta(q^2),\Delta(q\pi)\} = 2\Delta(q^2)\quad,\quad
 \{\Delta(q\pi),\Delta(\pi^2)\} = 2\Delta(\pi^2)\quad,\quad
 \{\Delta(q^2),\Delta(\pi^2)\} = 4\Delta(q\pi)\,.
\end{equation}
As our first canonical coordinate we choose
$s=\sqrt{\Delta(q^2)}$. Identifying the (negative) parameter along its
Hamiltonian flow with the new momentum $p_s$, we have the differential
equations
\begin{eqnarray}
 \frac{\partial\Delta(q^2)}{\partial p_s} &=&
 -\{\Delta(q^2),\sqrt{\Delta(q^2)}\}=0\\ 
 \frac{\partial\Delta(q \pi)}{\partial p_s} &=&
 -\{\Delta(q\pi),\sqrt{\Delta(q^2)}\}=  \sqrt{\Delta(q^2)}=s \label{Diff1}\\
 \frac{\partial\Delta(\pi^2)}{\partial p_s} &=&
 -\{\Delta(\pi^2),\sqrt{\Delta(q^2)}\}=2\frac{\Delta(q
   \pi)}{\sqrt{\Delta(q^2)}}= 2\frac{\Delta(q \pi)}{s}\,.  \label{Diff2}
\end{eqnarray}
Since $s$ is held constant in these equations, we can first solve
(\ref{Diff1}) by a simple integration,
\begin{equation}
 \Delta(q \pi)=s p_s +f_1(s)\,,
\end{equation}
insert the result in (\ref{Diff2}) and integrate once more:
\begin{equation}
 \Delta(\pi^2)=p_s^2+2\frac{f_1(s)}{s} p_s +f_2(s)\,.
\end{equation}
Computing $\{\Delta(q \pi),\Delta(\pi^2)\}$ using the canonical nature of the
variables $s$ and $p_s$, and requiring that it equal $2\Delta(\pi^2)$ implies
two equations:
\begin{equation}
 \frac{{\rm d}f_1}{{\rm d}s}=\frac{f_1}{s} \quad,\quad \frac{{\rm d}f_2}{{\rm
     d}s} = 2\frac{f_1}{s^2} \frac{{\rm d}f_1}{{\rm d}s} - 2\frac{f_2}{s} \,.
\end{equation}
They are solved by
\begin{equation}
 f_1(s)=U_2s \quad,\quad f_2(s)=\frac{U_1}{s^2}+U_2^2
\end{equation}
with constants $U_1$ and $U_2$. We can eliminate $U_2$ by a canonical
transformation replacing $p_s$ with $p_s+U_2$. The constant $U_1$ is the
Casimir coordinate. The resulting moments in terms of Casimir--Darboux
variables are
\begin{equation}
 \Delta(q^2) = s^2\quad,\quad
\Delta(q \pi) = s  p_s\quad,\quad
\Delta(\pi^2) = p_s^2+\frac{U_1}{s^2}
\end{equation}
as in (\ref{MomentCanReal}) or \cite{GaussianDyn,QHDTunneling}.  The Casimir
coordinate $U_1$ can be interpreted as the left-hand side of Heisenberg's
uncertainty relation,
\begin{equation} \label{U1}
 \Delta(q^2) \Delta(\pi^2) - \Delta(q \pi)^2 = U_1 \geq \frac{\hbar^2}{4}\,,
\end{equation}
which is a constant of motion at second semiclassical order.

\subsection{Poisson tensors of rank greater than two}

If we have a Poisson tensor of rank greater than two, we have to iterate the
procedure used in our example in order to find additional canonical pairs. In
general, it may be difficult to solve some of the differential equations
explicitly. 

Instead of using general solutions and eliminating surplus parameters through
canonical transformations, in practice it is more useful to make suitable
choices for functions such as $f_1$ and $f_2$ in the preceding example. There
are wrong choices in the sense that the procedure may terminate before the
required number of coordinates has been found, in which case one obtains a
non-faithful canonical realization. Usually, it is not difficult to see which
choices lead to a loss of degrees of freedom.

In order to iterate the procedure, we use the following method related to the
notion of Dirac observables in canonical relativistic systems
\cite{DirQuant,BergmannTime,PartialCompleteObs}. Having found a canonical pair
$(s,p_s)$ on a (sub)manifold of dimension $d$, we construct $d-2$ independent
functions $f_i$ such that $\{f_i,s\}=0=\{f_i,p_s\}$ for all $i$. These
functions are then Dirac observables with respect to $s$ and $p_s$. The
construction of Dirac observables is, in general, a very difficult task, and
in fact presents one of the main problems of canonical quantum gravity. Here,
however, the structure of already-constructed canonical coordinates helps to
make the construction of suitable $f_i$ feasible. In particular, the free
functions that remain after constructing $s$ and $p_s$, such as $f_1$ and
$f_2$ in the example, are, by construction, independent of $s$, and therefore
already fulfill $\{f_i,p_s\}=0$. 

Only a single set of conditions, $\{f_i,s\}=0$, then remains to be implemented
by suitable combinations of the original $f_i$, which can be done by
eliminating integration parameters in the flow $F_s(t)$. For instance, had we
not already known that $U_1$ in (\ref{U1}) is a Casimir function, we could
have derived it as follows: The flow generated by $s^2=\Delta(q^2)$ on the
remaining moments is determined by the differential equations
\begin{equation}
 \frac{{\rm d}\Delta(q\pi)}{{\rm d}t} = -2\Delta(q^2)
 \quad,\quad \frac{{\rm d}\Delta(\pi^2)}{{\rm d}t} = -4\Delta(q\pi)\,.
\end{equation}
The first equation implies that $\Delta(q\pi)[t]=-2\Delta(q^2)t+d$ with
$t$-independent $d$. Inserting this solution in the second equation, we find
$\Delta(\pi^2)[t]=4\Delta(q^2)t^2-4dt+e$ with another constant $e$. We now
eliminate $t$ by inserting $t=\frac{1}{2}(d-\Delta(q\pi)[t])/s^2$ in
$\Delta(\pi^2)[t]$: 
\begin{equation}
 \Delta(\pi^2)[t]= \frac{\Delta(q\pi)[t]^2}{\Delta(q^2)}-
 3\frac{d^2}{\Delta(q^2)}+e\,.
\end{equation}
Therefore, $U_1=\Delta(q^2)\Delta(\pi^2)[t]-\Delta(q\pi)[t]^2= -3d^2+es^2$ is
independent of $t$, which implies ${\rm d} U_1/{\rm
  d}t=\{U_1,\Delta(q^2)\}=0$, and $U_1$ is a Dirac observable with respect to
$\Delta(q^2)$ which can be used as a coordinate Poisson orthogonal to $s$.

The Poisson bracket of two Dirac observables is also a Dirac observable.
(This property may be useful for calculating further Dirac observables once
more than two have been found.)  Given a complete set of Dirac observables,
they form coordinates on a Poisson manifold, and we can compute their Poisson
brackets from their expressions in terms of the original variables. On this
new Poisson manifold, we proceed as in the first step, and then iterate. The
procedure terminates when we reach the full dimension, in which case the
Poisson manifold is symplectic, or when we obtain a complete set of Poisson
commuting Dirac observables. The commuting Dirac observables are the Casimir
functions. Because all coordinates constructed in this way are functions of
the original variables (the moments in our case of interest), there is no need
to transport coordinates to successive hypersurfaces.

\subsection{Second-order canonical realization for two classical degrees of
  freedom}

A non-trivial example of our general procedure is given by the second-order
semiclassical truncation of a system with two pairs of classical degrees of
freedom, $(q_1,\pi_1)$ and $(q_2,\pi_2)$. We obtain ten moments: two
fluctuations and one covariance for each pair, as well as four
cross-covariance such as $\Delta(q_1q_2)$. The rank of the resulting Poisson
tensor is eight, so that we should construct four canonical pairs and two
Casimir functions.

Since we already discussed the case of a single canonical pair, we can speed
up the first step and construct two canonical pairs at the same time by
defining $s_1=\sqrt{\Delta(q_1^2)}$ and $s_2=\sqrt{\Delta(q_2^2)}$.  Their
canonical momenta can be generated as in the case of a single degree of
freedom, but analogs of the functions $f_i$ could now depend on all the
remaining canonical variables: We have
\begin{equation} \label{Deltaqp1}
 \Delta(q_1\pi_1) = s_1p_1+f_{q_1\pi_1} \quad,\quad \Delta(\pi_1^2) =
 p_1^2+2\frac{p_1}{s_1}f_{q_1\pi_1}+ f_{q_1\pi_1}^2+\frac{f_{\pi_1^2}}{s_1^2}
\end{equation}
and
\begin{equation}\label{Deltaqp2}
 \Delta(q_2\pi_2) = s_2p_2+f_{q_2\pi_2} \quad,\quad \Delta(\pi_2^2) =
 p_2^2+2\frac{p_2}{s_2}f_{q_2\pi_2}+ f_{q_2\pi_2}^2+\frac{f_{\pi_2^2}}{s_2^2}
\end{equation}
with four functions $f_{q_1\pi_1}$, $f_{\pi_1^2}$, $f_{q_2\pi_2}$ and
$f_{\pi_2^2}$ independent of $s_1$, $p_1$, $s_2$ and $p_2$.

We now have to find spaces which are Poisson orthogonal to
$(s_1,p_1,s_2,p_2)$, or functions of the moments which Poisson commute with
all four canonical coordinates. If we choose $f_{q_1\pi_1}=0=f_{q_2\pi_2}$,
this condition is equivalent to having moments which Poisson commute with
$\Delta(q_1^2)$, $\Delta(q_1p_1)$, $\Delta(q_2^2)$ and $\Delta(q_2p_2)$. Two
such functions are
\begin{equation}
 f_{\pi_1^2}=s_1^2\Delta(\pi_1^2)-s_1^2p_1^2=
 \Delta(q_1^2)\Delta(\pi_1^2)-\Delta(q_1\pi_1)^2 =:f_1
\end{equation}
and
\begin{equation}
 f_{\pi_2^2}=s_2^2\Delta(\pi_2^2)-s_2^2p_2^2=
 \Delta(q_2^2)\Delta(\pi_2^2)-\Delta(q_2\pi_2)^2=:f_2
\end{equation}
obtained simply by solving (\ref{Deltaqp1}) and (\ref{Deltaqp2}) for
$f_{\pi_1^2}$ and $f_{\pi_2^2}$.  After computing the Poisson brackets between
all the cross-covariances and $\Delta(q_1^2)=s_1^2$,
$\Delta(q_1\pi_1)=s_1p_1$, $\Delta(q_2^2)=s_2^2$ and
$\Delta(q_2\pi_2)=s_2p_2$, we can construct a complete set of other Poisson
commuting functions by integrating flow equations generated by
$\Delta(q_1^2)$, $\Delta(q_1\pi_1)$, $\Delta(q_2^2)$ and $\Delta(q_2\pi_2)$.
The resulting combinations are
\begin{eqnarray}
 f_3 &=& \Delta(q_1\pi_2)\Delta(q_2\pi_1)-\Delta(q_1q_2)\Delta(\pi_1\pi_2)\\
 f_4 &=&
 \Delta(q_1^2)\frac{\Delta(q_2\pi_1)}{\Delta(q_1q_2)}-\Delta(q_1\pi_1)\\
 f_5 &=&
 \Delta(q_2^2)\frac{\Delta(q_1\pi_2)}{\Delta(q_1q_2)}-\Delta(q_2\pi_2)\\
 f_6 &=& \frac{\Delta(q_1^2)\Delta(q_2^2)}{\Delta(q_1q_2)^2}\,,
\end{eqnarray}
as can be checked explicitly.  The Poisson brackets between these six
functions are closed, so that we can iterate the procedure.

We start the next step by defining $s_3=f_6$, which is the inverse of the
squared correlation between the two particle positions. Its flow equations
impose conditions on derivatives of functions Poisson-commuting with $p_3$,
which can again be integrated. Solving some of the integrals, we obtain $p_3$
as a function of the $f_i$ and $s_3$, explicitly
\begin{equation} \label{p3}
 p_3=\frac{f_4+f_5}{4s_3(1-s_3)}\,.
\end{equation}
Moreover, the four combinations
\begin{eqnarray}
 g_1 &=& f_1+\frac{(f_4+f_5)^2}{4(1-f_6)}+
 \frac{1}{2}\frac{(f_4+f_5)(f_4-f_5)}{1-f_6}\\
 g_2 &=& f_2+\frac{(f_4+f_5)^2}{4(1-f_6)}-
 \frac{1}{2}\frac{(f_4+f_5)(f_4-f_5)}{1-f_6}\\
 g_3 &=& f_3+\frac{(f_4+f_5)^2}{4(1-f_6)}\\
g_4 &=& \frac{1}{2}(f_4-f_5)
\end{eqnarray}
Poisson commute with $s_3$ and $p_3$, as can again be checked explicitly. It
turns out that
\begin{equation} \label{Casimirg}
 g_1+g_2-2g_3=U_1
\end{equation}
is the quadratic Casimir of the full moment system. Using $U_1$, we have three
remaining variables, which can conveniently be chosen to be $g_1\pm g_2$ and
$g_4$. Their mutual Poisson brackets are again closed. 

The next step of the procedure leads to the combinations
\begin{eqnarray}
 h_1 &=& \frac{g_4}{\sqrt{s_3-1  }}\\
h_2 &=& (g_1-g_2)\sqrt{\frac{s_3-1}{s_3}}\\
h_3 &=& \frac{(1-s_3)(g_1+g_2)+s_3 U_1+2(1+s_3)(1-s_3)^{-1}g_4^2}{\sqrt{s_3}}
\end{eqnarray}
Poisson-commuting with $s_3$ and $p_3$, in addition to $U_1$.  We choose
$p_4=h_1$ as our final canonical momentum, such that invariance under its flow
implies 
\begin{eqnarray}
h_2&=&A(p_4) \cos(s_4)\\
h_3&=&A(p_4) \sin(s_4)
\end{eqnarray}
with some function $A(p_4)$. From the remaining Poisson brackets of $h_i$, it
follows that
\begin{equation}
A(p_4) \frac{{\rm d} A(p_4)}{{\rm d} p_4}=-8 p_4 U_1+32 p_4^3\,.
\end{equation}
The general solution of this equation is
\begin{equation}
A(p_4)=\sqrt{U_2-8p_4^2 U_1+16 p_4^4}
\end{equation}
with a constant of integration $U_2$ which can be interpreted as the second
Casimir. (At this point, it could be any function of the quadratic and quartic
Casimirs).

To summarize, we express the original moments in terms of Casimir--Darboux
variables.  For moments of the first classical pair of degrees of freedom, we
find
\begin{eqnarray}
\Delta(q_1^2)&=&s_1^2\quad,\quad
\Delta(q_1 \pi_1)= s_1 p_1 \label{TwoMoments1}\\
\Delta(\pi_1^2)&=&p_1^2+\frac{\Phi(s_3,p_3,s_4,p_4)}{s_1^2}
\end{eqnarray}
with
\begin{eqnarray}
\Phi(s_3,p_3,s_4,p_4)&=&-\frac{s_3+1}{s_3-1}p_4^2-4 s_3 \sqrt{s_3-1}p_3 p_4+4
s_3^2\left(s_3-1\right)p_3^2+\frac{1}{2}\frac{s_3}{s_3-1} U_1\\
&&-\frac{1}{2}\frac{\sqrt{s_3}}{s_3-1}\sqrt{U_2-8p_4^2 U_1+16
  p_4^4}\left(\sqrt{s_3-1}\cos{(s_4)}+\sin{(s_4)}\right)\,. \nonumber 
\end{eqnarray}
For moments of the second classical pair of degrees of freedom,
\begin{eqnarray}
\Delta(q_2^2)&=&s_2^2\quad,\quad
\Delta(q_2 \pi_2)= s_2 p_2\\
\Delta(\pi_2^2)&=&p_2^2+\frac{\Gamma(s_3,p_3,s_4,p_4)}{s_2^2}
\end{eqnarray}
with
\begin{eqnarray}
\Gamma(s_3,p_3,s_4,p_4)&=&-\frac{s_3+1}{s_3-1}p_4^2+4 s_3 \sqrt{s_3-1}p_3
p_4+4 s_3^2\left(s_3-1\right)p_3^2
+\frac{1}{2}\frac{s_3}{s_3-1} U_1\\
&&-\frac{1}{2}\frac{\sqrt{s_3}}{s_3-1}\sqrt{U_2-8p_4^2 U_1+16
  p_4^4}\left(-\sqrt{s_3-1}\cos{(s_4)}+\sin{(s_4)}\right)\,. \nonumber
\end{eqnarray}
Finally, we have
\begin{eqnarray}
\Delta(\pi_1 \pi_2)&=&\frac{p_1
  p_2}{\sqrt{s_3}}+\sqrt{\frac{s_3-1}{s_3}}
\left(\frac{p_2}{s_1}-\frac{p_1}{s_2}\right)p_4\\ 
&&-2
\sqrt{s_3}\left(s_3-1\right)
\left(\frac{p_1}{s_2}+\frac{p_2}{s_1}\right)p_3+\frac{\left(3
    s_3-1\right)}{s_1 s_2 \sqrt{s_3}\left(s_3-1\right)}p_4^2\nonumber\\ 
&&-4\frac{\left(s_3-1\right)s_3^{3/2}}{s_1 s_2}p_3^2-\frac{\sqrt{s_3}}{2 s_1
  s_2 \left(s_3-1\right)}U_1\nonumber\\ 
&&+\frac{s_3}{2 s_1 s_2 \left(s_3-1\right)}\sin{(s_4)}\sqrt{U_2-8p_4^2 U_1+16
  p_4^4}\nonumber\\ 
\Delta(q_1\pi_2)&=&\frac{p_2
  s_1}{\sqrt{s_3}}-\sqrt{\frac{s_3-1}{s_3}}\frac{s_1}{s_2}p_4-2
\left(s_3-1\right)\sqrt{s_3}\frac{s_1}{s_2}p_3\\ 
\Delta(q_2\pi_1)&=&\frac{p_1
  s_2}{\sqrt{s_3}}+\sqrt{\frac{s_3-1}{s_3}}\frac{s_2}{s_1}p_4-2
\left(s_3-1\right)\sqrt{s_3}\frac{s_2}{s_1}p_3\\ 
\Delta(q_1 q_2)&=& \frac{s_1 s_2}{\sqrt{s_3}} \label{TwoMoments2}
\end{eqnarray}	
for the cross-covariances.

\subsection{Third-order semiclassical truncation for single pair of
  degrees of freedom}

Third-order moments are subject to linear Poisson brackets within a
third-order truncation. In particular, the Poisson bracket of any pair of
third-order moments is zero within this truncation, and we have linear
brackets between second-order and third-order moments, such as
\begin{equation}
 \{\Delta(q^2),\Delta(q^2\pi)\}= 2\Delta(q^3)\quad,\quad
 \{\Delta(q^2),\Delta(q\pi^2)\}= 4\Delta(q^2\pi)\quad,\quad
 \{\Delta(q^2),\Delta(\pi^3)\}=6\Delta(q\pi^2)
\end{equation}
and so on. Thanks to the truncation, the brackets still define a linear Lie
algebra, but it is not semisimple because the third-order moments span an
Abelian ideal. This seven-dimensional Lie algebra is the semidirect product
${\rm sp}(2,{\mathbb R})\ltimes{\mathbb R}^4$ where ${\rm sp}(2,{\mathbb R})$,
spanned by the second-order moments, acts on ${\mathbb R}^4$, spanned by the
third-order moments, according to 
\begin{eqnarray}
 A&=&-\frac{1}{2}\Delta(\pi^2)=
 \left(\begin{array}{cccc}
     0&0&0&0\\3&0&0&0\\0&2&0&0\\0&0&1&0 \end{array}\right)\,, \\
B&=&\frac{1}{2}\Delta(q^2)=
\left(\begin{array}{cccc}
    0&1&0&0\\0&0&2&0\\0&0&0&3\\0&0&0&0 \end{array}\right)\,, \\
C&=&\Delta(q\pi)=
\left(\begin{array}{cccc} -3&0&0&0\\0&-1&0&0\\0&0&1&0\\0&0&0&3
\end{array}\right)
\end{eqnarray}
using (\ref{ABCDelta}). Computing the Casimir
\begin{equation}
K=-\frac{1}{2}(AB+BA)-\frac{1}{4}C^2=-\frac{15}{4}{\mathbb
  I}=-\frac{3}{2}\left(\frac{3}{2}+1\right){\mathbb I}\,,
\end{equation} 
this action is recognized as the spin-$3/2$ representation of ${\rm
  sp}(2,{\mathbb R})$.

Guided by our second-order examples, we make the choice
\begin{eqnarray}
\Delta(q^2)&=&s_1^2\\
\Delta(q \pi)&=& s_1 \, p_{1} 
\end{eqnarray}
as the first step in the introduction of canonical coordinates.
Suitable variables on the hypersurface Poisson orthogonal to
$(s_1,p_{1})$ are
\begin{eqnarray*}
\label{f}
f_1&=&\Delta(q^2)\Delta(\pi^2)-\Delta(q\pi)^2\\
f_2&=&\Delta(q^2)\frac{\Delta(q^2\pi)}{\Delta(q^3)}-\Delta(q\pi)\\
f_3&=&\frac{\Delta(q^2)^2}{\Delta(q^3)^2}\left(\Delta(q^2\pi)^2-
  \Delta(q\pi^2)\Delta(q^3) \right)\\
f_4 &=& 2 \Delta(q\pi)+\Delta(q^2)\frac{\Delta(q^3)\Delta(\pi^3)-
  \Delta(q\pi^2)\Delta(q^2\pi)}{\Delta(q^2\pi)^2-\Delta(q\pi^2)\Delta(q^3)} \, .
\end{eqnarray*}
The dimension of the Poisson manifold at third order is $D=7$, while the rank
of the Poisson tensor is six. We therefore expect three degrees of freedom and
one Casimir function. One additional coordinate Poisson commuting with
$(s_1,p_1)$ is needed to have seven independent variables.  Since the Poisson
brackets of $f_i$ are closed, the last variable Poisson commuting with
$(s_1,p_{1})$ has to be the Casimir function, which by ansatz can be found to
be
\begin{eqnarray}
U_1&=&4\left(\Delta(q\pi^2)^2-\Delta(q^2\pi)\Delta(\pi^3)\right)
\left(\Delta(q^2\pi)^2-\Delta(q^3)\Delta(q\pi^2)\right)\\
&-&\left(\Delta(q^2\pi)\Delta(q\pi^2)-\Delta(q^3)\Delta(\pi^3)\right)^2 \, .
\end{eqnarray}

To initiate the next step, we choose
\begin{equation}
s_2=f_3
\end{equation}
and integrate its flow equations. The resulting expressions tell us that
\begin{equation}
p_2=\frac{6 f_2+f_4}{16 s_2}\,,
\end{equation}
while
\begin{equation}
g_1=f_1+\frac{\left(6 f_2+f_4\right)^2}{16}\quad,\quad
g_2=-\frac{1}{2}f_2-\frac{1}{4}f_4 
\end{equation}
Poisson commute with $s_2$ but not with $p_2$.  After a further transformation
of variables, we obtain the remaining canonical pair
\begin{eqnarray}
s_3&=&\frac{g_2}{\sqrt{s_2}}\\
p_3&=&-\frac{2 g_1-7 s_2+10 p_3^2 s_2}{6 \sqrt{s_2}(-1+4 p_3^2)} \, ,
\end{eqnarray}
as can be checked directly.

The resulting faithful canonical realization is given by the second-order
moments 
\begin{eqnarray}
\Delta(\pi^2)&=&p_1^2+\frac{f_1(s_2,p_2,s_3,p_3)}{s_1^2}\\
\Delta(q\pi)&=& s_1 p_1\\
\Delta(q^2)&=&s_1^2
\end{eqnarray}
where
\begin{equation}
  f_1(s_2,p_2,s_3,p_3) = -3 \sqrt{s_2}\left(4
    s_3^2-1\right)p_3+\frac{1}{2}\left(7 -10
    s_3^2\right)s_2-16 s_2^2 p_2^2\,,
\end{equation}
and third-order moments
\begin{eqnarray}
\Delta(\pi^3)&=&
\frac{1}{\sqrt{s_2}s_1^3}\Phi(s_i,p_j)\left(\frac{U_1}{16 
    s_2 s_3^2-4 s_2}\right)^{1/4}\\ 
\Delta(q\pi^2)&=&\frac{1}{s_1 \sqrt{s_2}}\left(p_1 s_1
  +\left(s_3-1\right)\sqrt{s_2}+4 s_2 p_2\right)\\ 
&&\times\left(p_1 s_1 +\left(s_3+1\right)\sqrt{s_2}+4 s_2
  p_2\right)\left(\frac{U_1}{16 s_2 s_3^2-4 s_2}\right)^{1/4}\nonumber\\ 
\Delta(q^2\pi)&=&\frac{1}{\sqrt{s_2}}\left(p_1 s_1^2+s_1\left(p_3 \sqrt{s_2}+4 s_2
    p_2\right)\right)\left(\frac{U_1}{16 s_2 s_3^2-4 s_2}\right)^{1/4}\\ 
\Delta(q^3)&=&\frac{s_1^3}{\sqrt{s_2}}\left(\frac{U_1}{16 s_2 s_3^2-4
    s_2}\right)^{1/4}
\end{eqnarray}
with
\begin{eqnarray}
\Phi(s_i,p_j)&=&p_1^3 s_1^3+3 p_1^2 s_1^2 \sqrt{s_2}s_3
+ 3 p_1 s_1 s_2 \left(s_3^2+4 s_1p_1 p_2-1\right)+64 p_2^3 s_2^3\\
&&+s_2^{3/2}s_3\left(s_3^2+24 s_1p_1 p_2 -7\right)+48 p_2^2
s_2^{5/2}s_3
+12 p_2 s_2^2\left(s_3^2+4 s_1p_1 p_2-1\right)\,.\nonumber
\end{eqnarray}

\subsection{Momentum dependence}

In \cite{QHDTunneling}, the moments are quadratic in the new momentum
$p_s$. This property is useful because it implies an effective Hamiltonian
(\ref{Heff}) with standard kinetic term, quadratic in the classical momentum
$\pi$ (the expectation value) and the new momentum $p_s$ related to
$\Delta(\pi^2)$:
\begin{eqnarray}
 \langle\hat{H}\rangle &=& \frac{\langle\hat{\pi}^2\rangle}{2m}+V(\hat{q})=
 \frac{\pi^2+\Delta(\pi^2)}{2m}+ V(q)+ \frac{1}{2}V''(q) \Delta(q^2)+\cdots
 \nonumber \\
 &=&
 \frac{\pi^2}{2m}+ \frac{p_s^2}{2m}+ \frac{U}{2m}+
 V(q)+\frac{1}{2}V''(q)s^2+\cdots
\end{eqnarray}
The corresponding property for a bosonic realization implies that generators
of a Lie algebra have some terms bilinear in the boson variables. (However,
bosonic realizations corresponding to canonical realizations of moment
algebras cannot be completely bilinear, owing to Casimir terms such as
$U/s^2$.) Our third-order realization for a single classical degree of freedom
is similar in that $\Delta(\pi^2)$ is quadratic in the new momenta, altough
with $s$-dependent coefficients.

Unlike the example of a single pair of degrees of freedom, the moments for two
pairs of degrees of freedom, given so far, are not quadratic in the new
momenta. In fact, we can prove by ansatz that, for a second-order
semiclassical truncation for two classical degrees of freedom, there is no
faithful representation quadratic in momenta with $s$-independent
coefficients. The Poisson tensor has rank eight, so that we are looking for
four canonical pairs $(s_j,p_i)$ and two Casimir functions.

We write
\begin{eqnarray}
 \Delta(\pi_1^2) &=& p_1^2+p_3^2+ F_1(s_i)
 p_1+F_2(s_i)p_3+F(s_i)\\
  \Delta(\pi_2^2) &=& p_2^2+p_4^2+ G_1(s_i)
 p_2+G_2(s_i)p_4+G(s_i)\\
 \Delta(\pi_1\pi_2) &=& p_1p_2+p_3p_4+ H_1(s_i)p_1+H_2(s_i)p_2+H_3(s_i)p_3+
 H_4(s_i)p_4+ H_5(s_i)
\end{eqnarray}
and choose
\begin{equation}
 \Delta(q_1^2)=s_1^2+s_3^2\quad,\quad \Delta(q_2^2)=s_2^2+s_4^2\quad,\quad
 \Delta(q_1q_2)=s_1s_2+s_3s_4\,. 
\end{equation}
A realization of the entire algebra can be generated by taking Poisson
brackets: We can compute
\begin{equation}
 \Delta(\pi_1\pi_2)=
 \frac{1}{4}\left\{\left\{\Delta(q_1q_2),\Delta(\pi_2^2)\right\}, 
\Delta(\pi_1^2)\right\}
\end{equation}
and, given this moment,
\begin{equation}
 \Delta(q_1\pi_2)= \frac{1}{2}\{\Delta(q_1^2),\Delta(\pi_1\pi_2)\}\quad,\quad 
\Delta(q_2\pi_1)= \frac{1}{2}\{\Delta(q_2^2),\Delta(\pi_1\pi_2)\}\,.
\end{equation}
Finally, once we know these three moments, we compute
\begin{equation}
 \Delta(q_1\pi_1)+\Delta(q_2\pi_2)=
 \{\Delta(q_1q_2),\Delta(\pi_1\pi_2)\}\quad,\quad
 -\Delta(q_1\pi_1)+\Delta(q_2\pi_2)= \{\Delta(q_1\pi_2),\Delta(q_2\pi_1)\}
\end{equation}
from which $\Delta(q_1\pi_1)$ and $\Delta(q_2\pi_2)$ follow from linear
combinations. If $F_1=F_2=F_3=0$, $G_1=G_2=G_3=0$, and
$H_1=H_2=H_3=H_4=H_5=0$, we have a non-faithful realization because there are
no Casimir variables. We therefore have to find suitable functions depending
on two additional variables, $U_1$ and $U_2$, such that the required Poisson
brackets are realized.

Evaluating all Poisson brackets for consistency conditions, such as
$\{\Delta(\pi_1^2),\Delta(\pi_2^2)\}=0$, we find the following mapping:
\begin{eqnarray}
\Delta(q_1^2) &=& s_1^2+s_3^2\\
\Delta(q_1\pi_1) &=& s_1p_1+s_3p_3+
\frac{1}{2}s_1s_2U_1\left(\frac{1}{s_2^2}-\frac{1}{s_1^2}\right) +
\frac{1}{2}s_3s_4U_2\left(\frac{1}{s_4^2}-\frac{1}{s_3^2}\right)\\
\Delta(\pi_1^2) &=& p_1^2+p_3^2+
p_1s_2U_1\left(\frac{1}{s_2^2}-\frac{1}{s_1^2}\right)+
p_3s_4U_2\left(\frac{1}{s_4^2}-\frac{1}{s_3^2}\right)\\
&&+
\frac{1}{4}s_2^2U_1^2\left(\frac{1}{s_2^2}-\frac{1}{s_1^2}\right)^2+
\frac{1}{4}s_4^2U_2^2 \left(\frac{1}{s_4^2}-\frac{1}{s_3^2}\right)^2 \nonumber
\end{eqnarray}
for the first classical degree of freedom,
\begin{eqnarray}
 \Delta(q_2^2) &=& s_2^2+s_4^2\\
\Delta(q_2\pi_2) &=&
s_2p_2+s_4p_4+\frac{1}{2}s_1s_2U_1\left(\frac{1}{s_1^2}-\frac{1}{s_2^2}\right)+
\frac{1}{2}s_3s_4U_2\left(\frac{1}{s_3^2}-\frac{1}{s_4^2}\right)\\
\Delta(\pi_2^2) &=& p_2^2+p_4^2+
p_2s_1U_1\left(\frac{1}{s_1^2}-\frac{1}{s_2^2}\right)
+p_4s_3U_2\left(\frac{1}{s_3^2}-\frac{1}{s_4^2}\right)\\
&&+ \frac{1}{4}s_1^2U_1^2
\left(\frac{1}{s_1^2}-\frac{1}{s_2^2}\right)^2+ \frac{1}{4}s_3^2U_2^2
\left(\frac{1}{s_3^2}-\frac{1}{s_4^2}\right)^2\nonumber
\end{eqnarray}
for the second classical degree of freedom, and
\begin{eqnarray}
\Delta(q_1q_2) &=& s_1s_2+s_3s_4\\
\Delta(q_1\pi_2) &=&
s_1p_2+s_3p_4+\frac{1}{2}s_1^2U_1\left(\frac{1}{s_1^2}-\frac{1}{s_2^2}\right)+
\frac{1}{2}s_3^2U_2 \left(\frac{1}{s_3^2}-\frac{1}{s_4^2}\right)\\
\Delta(q_2\pi_1) &=&
s_2p_1+s_4p_3+\frac{1}{2}s_2^2U_1\left(\frac{1}{s_2^2}-\frac{1}{s_1^2}\right)+ 
\frac{1}{2}s_4^2U_2 \left(\frac{1}{s_4^2}-\frac{1}{s_3^2}\right)\\
\Delta(\pi_1\pi_2) &=& p_1p_2+p_3p_4+ \frac{1}{2}p_1s_1U_1
\left(\frac{1}{s_1^2}-\frac{1}{s_2^2}\right)+ \frac{1}{2}p_2s_2U_1
\left(\frac{1}{s_2^2}-\frac{1}{s_1^2}\right)\\
&&+
\frac{1}{2}p_3s_3U_2\left(\frac{1}{s_3^2}-\frac{1}{s_4^2}\right)+
\frac{1}{2}p_4s_4U_2 \left(\frac{1}{s_4^2}-\frac{1}{s_3^2}\right)\nonumber\\
&&-
\frac{1}{4}s_1s_2U_1^2
\left(\frac{1}{s_2^2}-\frac{1}{s_1^2}\right)^2-\frac{1}{4}s_3s_4U_2^2
\left(\frac{1}{s_4^2}-\frac{1}{s_3^2}\right)^2\nonumber
\end{eqnarray}
for the cross-covariances. 

If the two free parameters $U_1$ and $U_2$ were
independent Casimir functions, we would have a faithful canonical realization.
However, the rank of the Jacobian of the transformation from $(s_i,p_j,U_I)$
to the moments can be seen to equal seven, and therefore the realization is
not faithful.  Moreover, the quadratic Casimir of the algebra,
\begin{equation}\label{casimir2}
C_2=\tr{\left(((\tau \Delta)^2)\right)}\,,
\end{equation}
can be computed explicitly and does not equal a function of $U_1$ and $U_2$
--- it depends on the coordinates as well. If the
map were faithful, we would have
\begin{equation}
\frac{\partial C_2}{\partial s_i}=\frac{\partial C_2}{\partial p_j}=0\,.
\end{equation}

Finally, we note that the canonical
transformation
\begin{eqnarray*}
P_1&=& p_1+\frac{1}{2}s_2
U_1\left(\frac{1}{s_2^2}-\frac{1}{s_1^2}\right)\quad,\quad 
P_2= p_2+\frac{1}{2}s_1 U_1\left(\frac{1}{s_1^2}-\frac{1}{s_2^2}\right)\\
P_3&=& p_3+\frac{1}{2}s_4
U_2\left(\frac{1}{s_4^2}-\frac{1}{s_3^2}\right)\quad,\quad 
P_4= p_4+\frac{1}{2}s_3 U_2\left(\frac{1}{s_3^2}-\frac{1}{s_4^2}\right)
\end{eqnarray*}
and $S_i=s_i$ maps our realization to the non-faithful
\begin{eqnarray*}
&&\Delta(q_1^2) = S_1^2+S_3^2\quad,\quad
\Delta(q_1\pi_1) = S_1P_1+S_3P_3\quad,\quad
\Delta(\pi_1^2) = P_1^2+P_3^2\\
&& \Delta(q_2^2) = S_2^2+S_4^2\quad,\quad
\Delta(q_2\pi_2) = S_2P_2+S_4P_4\quad,\quad
\Delta(\pi_2^2) = P_2^2+P_4^2\\
&&\Delta(q_1q_2) = S_1S_2+S_3S_4\quad,\quad
\Delta(q_1\pi_2) = S_1P_2+S_3P_4\\
&&\Delta(q_2\pi_1) =S_2P_1+S_4P_3\quad,\quad
\Delta(\pi_1\pi_2) = P_1P_2+P_3P_4\,,
\end{eqnarray*}
in which there are no free parameters that could play the role of Casimir
functions.  The only possibilities are therefore realizations non-quadratic in
momenta, or with non-standard, $s$-dependent kinetic terms. None of these
options can lead to a bilinear bosonic realization.

\subsection{Realizations of ${\rm sp}(2n,{\mathbb R})$}

The isomorphism between second-order semiclassical truncations and ${\rm
  sp}(2n,{\mathbb R})$ implies that faithful bosonic realizations of ${\rm
  sp}(4,{\mathbb R})$ cannot be bilinear in the boson variables. This result
underlines some of the difficulties in finding such realizations pointed out
in \cite{Bosonsp4,BosSymp}.  Given the generators $A_{ij}$ ($i\leq j$),
$B_{ij}$ ($i\leq j$) and $C_{ij}$, $i,j=1,\ldots,N$, of ${\rm sp}(2N,{\mathbb
  R})$ with relations (\ref{ABCRel}), it is easy to see that an explicit
isomorphism between ${\rm sp}(2N,{\mathbb R})$ and the second-order
semiclassical truncation with $N$ classical degrees of freedom is given by
\begin{equation}
 A_{ij} = \Delta(\pi_i\pi_j)\quad,\quad B_{ij}= \Delta(q_iq_j)\quad,
\quad C_{ij}=\Delta(q_i\pi_j)\,.
\end{equation}

In particular, for ${\rm sp}(4,{\mathbb R})$ we obtain a realization from
(\ref{TwoMoments1})--(\ref{TwoMoments2}) with four bosonic variables
$b_1=\frac{1}{\sqrt{2}}(s_1+ip_1)$, $b_2=\frac{1}{\sqrt{2}}(s_2+ip_2)$,
$b_3=\frac{1}{\sqrt{2}}(s_3+ip_3)$ and $b_4=\frac{1}{\sqrt{2}}(s_4+ip_4)$, in
addition to two Casimir variables $U_1$ and $U_2$. We do not reproduce here all
generators obtained by substituting bosonic variables in
(\ref{TwoMoments1})--(\ref{TwoMoments2}), but note that the resulting
expressions are rather different from the non-faithful form
(\ref{NonFaithful}). Even the moments that are bilinear in bosonic variables,
such as $B_{11}=s_1^2=\frac{1}{2}(b_1+b_1^*)^2$ or $C_{11}=s_1p_1=\frac{1}{2}i
\left((b_1^*)^2-b_1^2\right)$, depend on different combinations of the
$b_i$. These changes are required to maintain the reality conditions implied
by a bosonic realization. Moreover, our realization brings in the two Casimir
variables $U_1$ and $U_2$ in a way that requires a non-bilinear realization.

\section*{Acknowledgements}

This work was supported in part by NSF grant PHY-1607414.


\begin{thebibliography}{10}

\bibitem{QHDTunneling}
O.\ Prezhdo,
\newblock Quantized Hamiltonian Dynamics,
\newblock {\em Theor.\ Chem.\ Acc.} 116 (2006) 206

\bibitem{ROPP}
M.\ Bojowald,
\newblock Quantum cosmology: a review,
\newblock {\em Rep.\ Prog.\ Phys.} 78 (2015) 023901, [arXiv:1501.04899]

\bibitem{QHDHigher}
E.\ Pahl and O.\ Prezhdo,
\newblock Extension of quantized Hamilton dynamics to higher orders,
\newblock {\em J.\ Chem.\ Phys.} 116 (2002) 8704--8712

\bibitem{EffAc}
M.\ Bojowald and A.\ Skirzewski,
\newblock Effective Equations of Motion for Quantum Systems,
\newblock {\em Rev.\ Math.\ Phys.} 18 (2006) 713--745, [math-ph/0511043]

\bibitem{Karpacz}
M.\ Bojowald and A.\ Skirzewski,
\newblock Quantum Gravity and Higher Curvature Actions,
\newblock {\em Int.\ J.\ Geom.\ Meth.\ Mod.\ Phys.} 4 (2007) 25--52,
  [hep-th/0606232],
\newblock Proceedings of ``Current Mathematical Topics in Gravitation and
  Cosmology'' (42nd Karpacz Winter School of Theoretical Physics), Ed.\
  Borowiec, A.\ and Francaviglia, M.

\bibitem{CW}
M.\ Bojowald and S.\ Brahma,
\newblock Canonical derivation of effective potentials, [arXiv:1411.3636]

\bibitem{GaussianDyn}
F.\ Arickx, J.\ Broeckhove, W.\ Coene, and P.\ van Leuven,
\newblock Gaussian Wave-packet Dynamics,
\newblock {\em Int.\ J.\ Quant.\ Chem.: Quant.\ Chem.\ Symp.} 20 (1986)
  471--481

\bibitem{HP}
T.\ Holstein and H.\ Primakoff,
\newblock Field dependence of the intrinsic domain magnetization of a
  ferromagnet,
\newblock {\em Phys.\ Rev.} 58 (1940) 1098--1113

\bibitem{Nuclearsp6}
G.\ Rosensteel and D.~J.\ Rowe,
\newblock Nuclear ${\rm sp}(3,R)$ model,
\newblock {\em Phys.\ Rev.\ Lett.} 38 (1977) 10--14

\bibitem{AlgebraicCollective}
G.\ Rosensteel and D.~J.\ Rowe,
\newblock On the algebraic formulation of collective models III. The symplectic
  shell model of collective motion,
\newblock {\em Ann.\ Phys.} 126 (1980) 343--370

\bibitem{DynCollective}
J.\ Deenen and C.\ Quesne,
\newblock Dynamical group of microscopic collective states. I. One-dimensional
  case,
\newblock {\em J.\ Math.\ Phys.} 23 (1982) 878--889

\bibitem{CoherentSymp}
D.~J.\ Rowe,
\newblock Coherent state theory of the noncompact symplectic group,
\newblock {\em J.\ Math.\ Phys.} 25 (1984) 2662--2671

\bibitem{Bosonsp4}
O.\ Casta\~{n}os, E.\ Chac\'on, M.\ Moshinsky, and C.\ Quesne,
\newblock Boson realization of sp(4). I. The matrix formulation,
\newblock {\em J.\ Math.\ Phys.} 28 (1985) 2107--2123

\bibitem{BosSymp}
M.\ Moshinsky,
\newblock Boson realization of symplectic algebras,
\newblock {\em J.\ Phys.\ A} 18 (1985) L1--L6

\bibitem{Mukunda1}
N.\ Mukunda,
\newblock Dynamical symmetries and classical mechanics,
\newblock {\em Phys.\ Rev.} 155 (1967) 1383--1386

\bibitem{Mukunda2}
N.\ Mukunda,
\newblock Realizations of Lie algebras in classical mechanics,
\newblock {\em J.\ Math.\ Phys.} 8 (1967) 1069--1072

\bibitem{Mukunda3}
P.\ Chand, C.~L.\ Mehta, N.\ Mukunda, and E.~C.~G.\ Sudarshan,
\newblock Realization of Lie algebras by analytic functions of generators of a
  given Lie algebra,
\newblock {\em J.\ Math.\ Phys.} 8 (1967) 2048--2059

\bibitem{Rosen}
J.\ Rosen,
\newblock On realizations of Lie algebras and symmetries in classical and
  quantum mechanics,
\newblock {\em Il Nuovo Cim.} IL A (1967) 614--621

\bibitem{Subreps}
M.\ Iosifescu and H.\ Scutaro,
\newblock Poisson bracket realizations of Lie algebras and subrepresentations
  of $({\rm ad}^{\otimes k})_s$,
\newblock {\em J.\ Math.\ Phys.} 25 (1984) 2856--2962

\bibitem{LocalQuant}
R.\ Haag,
\newblock {\em Local Quantum Physics},
\newblock Springer-Verlag, Berlin, Heidelberg, New York, 1992

\bibitem{Counting}
A.\ Tsobanjan,
\newblock Semiclassical states on Lie algebras,
\newblock {\em J.\ Math.\ Phys.} 56 (2015) 033501, [arXiv:1410.0704]

\bibitem{HigherMoments}
M.\ Bojowald, D.\ Brizuela, H.~H.\ Hernandez, M.~J.\ Koop, and H.~A.\
  Morales-T\'ecotl,
\newblock High-order quantum back-reaction and quantum cosmology with a
  positive cosmological constant,
\newblock {\em Phys.\ Rev.\ D} 84 (2011) 043514, [arXiv:1011.3022]

\bibitem{Weinstein}
A.\ Cannas~da Silva and A.\ Weinstein,
\newblock {\em Geometric models for noncommutative algebras}, volume~10 of {\em
  Berkeley Mathematics Lectures},
\newblock Am.\ Math.\ Soc., Providence, 1999

\bibitem{Collective}
S.\ Goshen and H.~J.\ Lipkin,
\newblock A simple independent-particle system having collective properties,
\newblock {\em Ann.\ Phys.} 6 (1959) 301--309

\bibitem{Arnold}
V.~I.\ Arnold,
\newblock {\em Mathematical Methods of Classical Mechanics},
\newblock Springer, 1997

\bibitem{DirQuant}
P.~A.~M.\ Dirac,
\newblock {\em Lectures on Quantum Mechanics},
\newblock Yeshiva Press, 1969

\bibitem{BergmannTime}
P.~G.\ Bergmann,
\newblock Observables in General Relativity,
\newblock {\em Rev.\ Mod.\ Phys.} 33 (1961) 510--514

\bibitem{PartialCompleteObs}
B.\ Dittrich,
\newblock Partial and Complete Observables for Hamiltonian Constrained Systems,
\newblock {\em Gen.\ Rel.\ Grav.} 39 (2007) 1891--1927, [gr-qc/0411013]

\end{thebibliography}

\end{document}